\pdfoutput=1
\documentclass[aps,prl,twocolumn,10pt]{revtex4-1}

\usepackage[T1]{fontenc}
\fontencoding{T1}  
\usepackage[utf8]{inputenc}

\usepackage{enumitem}
\usepackage{tikz}
\usepackage{amsfonts}
\usepackage{amsmath,amssymb,amsthm}
\usepackage{graphicx}
\usepackage{fancyhdr}
\usepackage[breakable]{tcolorbox}
\usepackage{bbm}
\usepackage{url}
\usepackage{mathtools}
\usepackage{braket}
\usepackage{upgreek }
\usepackage{dsfont}
\usepackage{collectbox}
\usepackage{braket}
\usepackage{quantikz}

\usepackage[plain]{fancyref} %
\usepackage[ruled]{algorithm2e}

\usepackage{algpseudocode}

\usepackage{hyperref}

\newtheorem{thm}{Theorem}
\newtheorem*{thm*}{Theorem}
\makeatletter
\newcommand{\setthmtag}[1]{%
  \let\oldthethm\thethm%
  \newcommand{\thethm}{#1}%
  \g@addto@macro\endthm{%
    \addtocounter{thm}{-1}%
    \global\let\thethm\oldthethm}%
  }
\makeatother

\newtheorem{prop}[thm]{Proposition}
\newtheorem*{prop*}{Proposition}
\newtheorem{lemma}[thm]{Lemma}
\newtheorem*{lemma*}{Lemma}
\newtheorem{cor}[thm]{Corollary}
\newtheorem*{cor*}{Corollary}

\newtheorem*{cj*}{Conjecture}

\newtheorem*{Def*}{Definition}

\theoremstyle{definition}

\newtheorem*{rem*}{Remark}

\def\beq{\begin{equation}}
\def\eeq{\end{equation}}
\def\bq{\begin{quote}}
\def\eq{\end{quote}}
\def\ben{\begin{enumerate}}
\def\een{\end{enumerate}}
\def\bit{\begin{itemize}}
\def\eit{\end{itemize}}

\def\lb{\left(}
\def\rb{\right)}

\def\r|{\right|}

\newcommand\cN{\mathcal{N}}
\newcommand\cM{\mathcal{M}}
\newcommand\cP{\mathcal{P}}
\newcommand\cA{\mathcal{A}}

\newcommand{\cQ}{\mathcal{Q}}
\newcommand{\ketbra}[2]{|#1\rangle\langle #2|}
\newcommand{\tr}[1]{\operatorname{tr}\lb#1\rb}

\newcommand{\id}{\text{id}}

\newcommand{\cC}{\mathcal{C}}

\newcommand{\cT}{\mathcal{T}}
\newcommand\be{\begin{equation}}
\newcommand\ee{\end{equation}}

\begin{document}
\title{Efficient learning of the structure and parameters of local Pauli noise channels}

\author{\begingroup
\hypersetup{urlcolor=navyblue}
\href{https://orcid.org/0000-0001-7712-6582}{Cambyse Rouz\'{e}
\endgroup}
}
\email[Cambyse Rouz\'{e} ]{cambyse.rouze@tum.de}
 \affiliation{Zentrum Mathematik, Technische Universit\"{a}t M\"{u}nchen, 85748 Garching, Germany}

\author{\begingroup
\hypersetup{urlcolor=navyblue}
\href{https://orcid.org/0000-0001-9699-5994}{Daniel Stilck Fran\c{c}a}
\endgroup}
\affiliation{Univ Lyon, ENS Lyon, UCBL, CNRS, Inria, LIP, F-69342, Lyon Cedex 07, France}
\email[Daniel Stilck Fran\c ca ]{daniel.stilck\_franca@ens-lyon.fr}

\begin{abstract}
  The unavoidable presence of noise is a crucial roadblock for the development of large-scale quantum computers and the ability to characterize quantum noise reliably and efficiently with high precision is essential to scale quantum technologies further. Although estimating an arbitrary quantum channel requires exponential resources, it is expected that physically relevant noise has some underlying local structure, for instance that errors across different qubits have a conditional independence structure. Previous works showed how it is possible to estimate Pauli noise channels with an efficient number of samples in a way that is robust to state preparation and measurement errors, albeit departing from a known conditional independence structure.

  We present a novel approach for learning Pauli noise channels over n qubits that addresses this shortcoming. Unlike previous works that focused on learning coefficients with a known conditional independence structure, our method learns both the coefficients and the underlying structure. We achieve our results by leveraging a groundbreaking result by Bresler for efficiently learning Gibbs measures and obtain an optimal sample complexity of $\mathcal{O}(\log(n))$ to learn the unknown structure of the noise acting on $n$ qubits. This information can then be leveraged to obtain a description of the channel that is close in diamond distance from $\mathcal{O}(\textrm{poly}(n))$ samples. Furthermore, our method is efficient both in the number of samples and postprocessing without giving up on other desirable features such as SPAM-robustness, and only requires the implementation of single qubit Cliffords. In light of this, our novel approach enables the large-scale characterization of Pauli noise in quantum devices under minimal experimental requirements and assumptions.
\end{abstract}

\maketitle

\section{Introduction}
With the increasing size and quality of quantum devices in the last years, it becomes increasingly important to develop tools that can characterize and identify multiple (possibly correlated) error rates. One particular class of channels that has received considerable attention over the recent years is that of Pauli channels~\cite{fawzi2023lower,flammia2020efficient,chen2022quantum,flammia2021pauli,flammia2020efficient,PhysRevLett.130.200601}. This is justified by the fact that randomized compiling techniques can bring the noise affecting a quantum device into this normal form~\cite{PhysRevA.94.052325} and that it models well logical errors~\cite{PhysRevLett.130.200601}. Furthermore, Pauli channels have much more mathematical structure than arbitrary quantum channels, making it easier to treat them analytically. Unfortunately, even the simpler class of Pauli channels requires a number of samples that scales exponentially in the number of qubits~\cite{fawzi2023lower} to be characterized globally, i.e. in the diamond norm. 

However, it is expected that noise in physical systems will exhibit some locality, i.e.~the probability of different Pauli errors exhibits some conditional independence structure, as illustrated and explained further in Fig.~\ref{Fig1}. In this case, given a hypergraph that models the conditional independence structure of different Pauli errors and assuming that the size of the largest hyperedge is constant, the authors of \cite{harper2020efficient,flammia2020efficient,harper2023learning} showed how to learn the quantum channel with a number of samples that scales polynomially in the system's size. Furthermore, they can achieve this even in a way that is robust to state preparation and measurement errors (SPAM) and without using auxiliary qubits. However, these works need to assume that the conditional independence structure is known in advance, which limits the application of their methods in practice. In this work, we address and solve this bottleneck in estimating Pauli noise by devising a learning protocol that is both computationally efficient and can learn the conditional independence structure. Furthermore, our method shares the desirable features of previous protocols, namely efficient sample complexity and robustness to SPAM.

\begin{figure}
\centering
\includegraphics[scale=0.1]{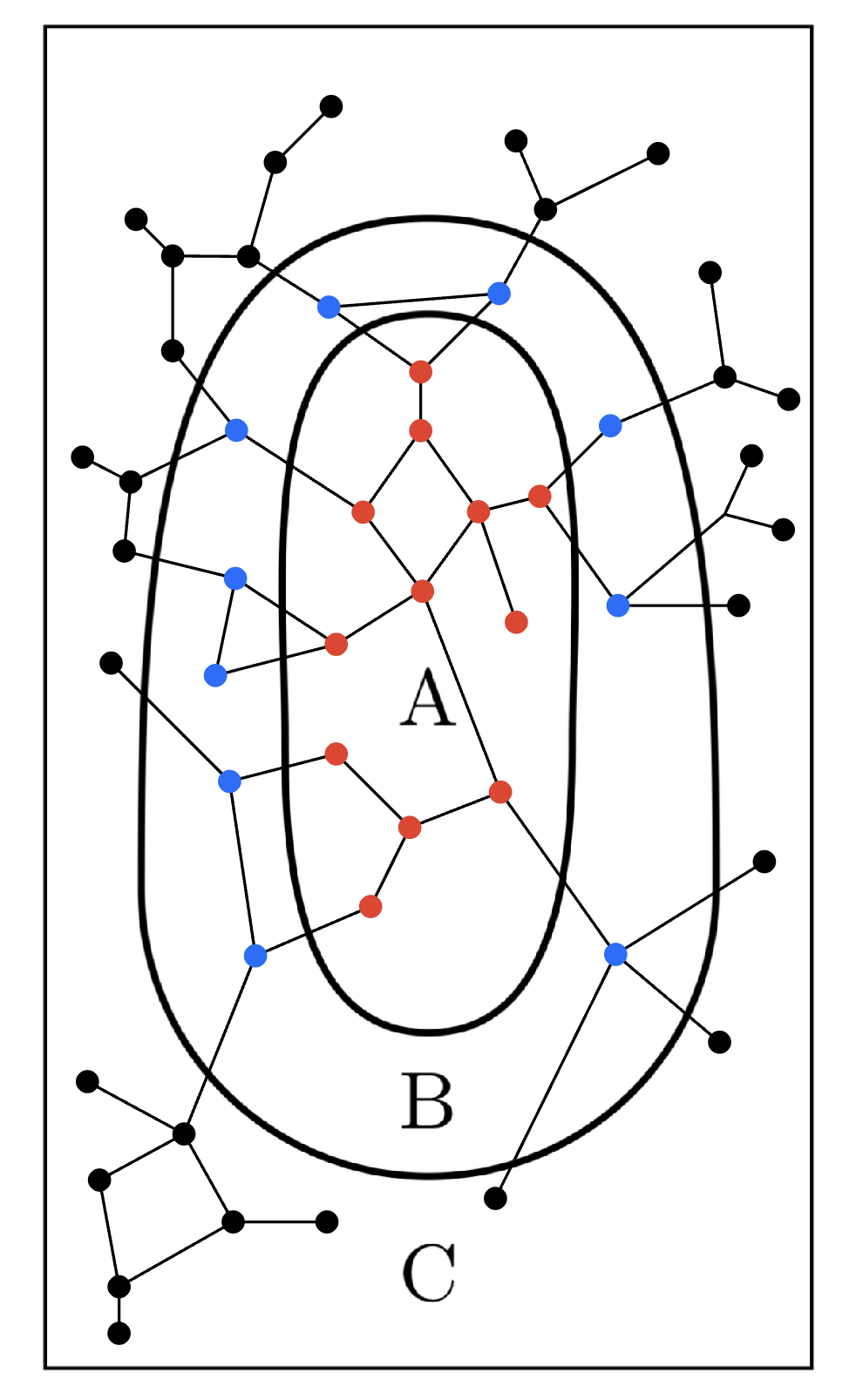}
\caption{The region $B$ of qubits (in blue) shields regions $A$ (in red) and $C$ (in black) away from each other. A measure $\mu$ satisfies conditional independence if $\mu(P_A,P_C|P_B)=\mu(P_A|P_B)\mu(P_C|P_B)$ for any such regions. In terms of Pauli errors, it means that given the errors that occurred in $B$, the distributions of errors in $A$ and $C$ is independent.}\label{Fig1}
\end{figure}

To achieve our results, we draw upon recent results on learning Markovian fields~\cite{bresler2015efficiently} and adapt them to our setting. The main technical difficulty we overcome is that these classical works assume access to samples from the underlying distribution. In contrast, without the use of auxiliary qubits it is not possible to sample from the probability distribution of various Pauli errors. To overcome this difficulty, we develop a new method that combines features from randomized benchmarking protocols~\cite{francca2018approximate,helsen2022general,helsen2019new,magesan2012characterizing,heinrich2023randomized} together with shadow process tomography protocols~\cite{franca2022efficient,kunjummen2023shadow,levy2021classical}, and which only requires the implementation of single qubit Cliffords. Akin to process shadows, our protocol allows for the estimation of all marginals of Pauli errors of a given size with a number of samples that scales logarithmically in the number of quits and exponentially in the size of the marginals. And like randomized benchmarking, this is done in a SPAM-robust way.
We then show how to run classical algorithms for learning Markovian fields solely with estimates of the marginals.

As our protocol checks all boxes for a practical and effective Pauli learning algorithm (no auxiliary systems required, SPAM-robustness, only requires single qubit Cliffords, optimal sample and efficient computational complexity and no knowledge of conditional independence structure required), we believe it will find widespread applicability to characterize large scale quantum devices.

\section{Preliminaries}
Let $d=2^n$ be the dimension of an $n$-qubit system. We use the notation $[d] := \{1,\dots,d\}$. A quantum state is a positive semi-definite Hermitian matrix of trace $1$. We denote the identity matrix by $\mathds{I}_d\in\mathds{C}^{d\times d}$ and by $\id_d: \mathds{C}^{d\times d}\rightarrow \mathds{C}^{d\times d}$ the identity map, and omit the subscript if the dimension is clear from context. A quantum channel is a map $\cN: \mathds{C}^{d\times d}\rightarrow \mathds{C}^{d\times d}$ of the form $\cN(\rho)=\sum_{k}A_k \rho A_k^\dagger$ where the Kraus operators $\{A_k\}_k$ satisfy $\sum_k A_k^\dagger A_k=\mathds{I}$. If the quantum channel $\cN$ satisfies further $\cN(\mathds{I})=\mathds{I}$, it is called \textit{unital}.

We define the diamond distance between two quantum channels $\cN$ and $\cM$ as the diamond norm of their difference:
\begin{align*}
    \|\cN-\cM\|_\diamond= \max_{\phi: \tr{\ketbra{\phi}{\phi}}=1}\|\id_d\otimes(\cN-\cM)(\ket{\phi}\bra{\phi}) \|_1
\end{align*}
where the the Schatten $1$-norm of a matrix $M$ is defined as $\|M\|_1=\operatorname{tr}(M^\dagger M)$.

Pauli channels are a special class of quantum channels whose Kraus operators are weighted Pauli operators. Formally, a Pauli quantum channel $\cP$ can be written as follows:
\begin{align}\label{equ:pauli_channel}
    \cP(\rho)= \sum_{P\in \mathds{P}_n} \mu(P) \,P \rho P
\end{align}
where we recall that the Pauli matrices $X,Y,Z$ take the form
\begin{equation}
X=\begin{pmatrix}
0 & 1 \\
1& 0
\end{pmatrix},\quad Y=\begin{pmatrix}
0 & -i \\
i & 0
\end{pmatrix},\quad Z=\begin{pmatrix}
1 & 0 \\
0 & -1
\end{pmatrix} \,,
\end{equation}
and where $\mu$ is a probability distribution over the set $\mathds{P}_n= \{\mathds{I},X,Y,Z\}^{\otimes n}$ of $n$-fold tensor products of Pauli matrices (or Pauli strings).

 The elements of $\mathds{P}_n$ either commute or anticommute: let $P$ and $Q$ be two Pauli operators, we have $PQ=(-1)^{P.Q}QP$, where $P.Q\coloneqq 0$ if $[P,Q]=0$ and $P.Q\coloneqq 1$ otherwise. Furthermore, we will denote the number of non-identity elements of a Pauli string $P$, its weight, by $w(P)$.

We consider the Pauli channel tomography problem which consists of learning a Pauli channel in the diamond norm. Given a precision  parameter $\epsilon>0$, the goal is to construct a Pauli channel $\widehat{\cP}$ satisfying with at least a probability $2/3$:
\begin{align*}
    \|\cP-\widehat{\cP}\|_\diamond\le \epsilon.
\end{align*}
An algorithm $\cA$ is $1/3$-correct for this problem if it outputs a Pauli channel $\epsilon$-close to $\cP$ with a probability of error at most $1/3$.
We choose to learn in the diamond norm because it characterizes the minimal error probability to distinguish between two quantum channels  when auxiliary systems are allowed \cite{watrous2018theory}. Since the diamond norm between two Pauli channels is exactly twice the total variation distance between their corresponding probability distributions \cite{magesan2012characterizing}, approximating the Pauli channel $\cP$ in diamond norm is equivalent to approximating the probability distribution $\mu$ in total variation distance. The latter is defined for two probability distributions $p$ and $q $ on $[d]$ as follows:
\begin{align*}
    \|p-q\|_{\operatorname{TV}}=\frac{1}{2}\sum_{i=1}^d |p_i-q_i|.
\end{align*}

The learner can only extract classical information from the unknown Pauli channel $\cP$ by performing a measurement on the output state for a choice of the input state. Throughout the paper, we only consider unentangled or incoherent measurements. That is, the learner can only measure with an $n$-qubit measurement device, and using auxiliary qubits or measuring multiple copies at once is not allowed. Furthermore, we only require product Pauli eigenstates as input states and measurements in product Pauli bases. This corresponds to a minimal and feasible experimental setup.

It is possible to identify the Pauli strings modulo phases with the space $(\mathbb{Z}_2\times \mathbb{Z}_2)^n$~\cite{watrous2018theory}. Through this identification, it is possible to see that for some Pauli string $Q$
\begin{equation}\label{eq:Fourier}
\alpha_Q\coloneqq \frac{1}{2^n}\tr{ Q\,\mathcal{P}(Q)}=\sum_{P\in\mathbb{P}_n}\,(-1)^{P.Q}\,\,\mu(P)
\end{equation}
is nothing but the Fourier coefficient of the distribution. Thus, it is possible to estimate the marginals of the distribution $\mu$ from such Fourier coefficients, as we will explain in more detail below.
In the appendix, we provide a simple scheme to learn these coefficients in a way that is robust to both state preparation and measurement (SPAM) errors and is inspired by randomized benchmarking (RB) combined with recent classical process shadows (PC) protocols~\cite{kunjummen2023shadow,levy2021classical,franca2022efficient}. It combines the best of these two types of protocols: SPAM-robustness (a feature of RB, but not of PC) and parallelization of estimates (a feature of PC, but not of RB).
For our protocol, we assume that we can prepare product Pauli eigenstates and measure in Pauli bases, and that the SPAM error is independent of the state and measurement and modeled by quantum channels $\Phi_1,\Phi_2$ (see Fig.\ref{fig:robust_parallelization_noisy}). Under these conditions, we have:
\begin{prop}[see Corollary \ref{ref:corSPAMlearning}]
Let $P_1,\ldots,P_m$ be $m$ Pauli strings such that $w(P_i)\leq w$, denote by $p_0$ the probability that no error occurs, i.e.~$p_0=\mu(\mathds{I})$, and  suppose that for a constant $C_{\operatorname{SPAM}}$ we have for all $1\leq i\leq m$:
\begin{align}
2^{-2n}|\tr{\Phi_1^*(P_i)P_i}\tr{\Phi_2^*(P_i)P_i}|\geq C_{\operatorname{SPAM}}.
\end{align}
Then $\mathcal{O}(3^w\epsilon^{-2}C_{\operatorname{SPAM}}^{-2}\log(m\delta^{-1}))$ state preparation and measurements in the setting described above suffice to estimate all $\alpha_{P_i}$ up to an error $\epsilon(1-p_0)$ with probability of success at least $1-\delta$. Furthermore, the maximal number $k$ of uses of the channel needed in each run satisfies $k=\mathcal{O}(\log(1/(1-p_0))$.
\end{prop}

From Fourier coefficients, one can easily reconstruct marginals of the distribution $\mu$ as follows: denoting by $A\subset [n]$ a certain set of qubits, the marginal $\mu_A$ is defined as 
\begin{equation*}
\mu_A(Q)\coloneqq \sum_{P\in \mathds{P}_{\overline{A}}}\,\mu(Q\otimes P)\,=\sum_{P'\in \mathds{P}_{\overline{A}}}(-1)^{Q\cdot P'}\alpha_{P'},Q\in\mathds{P}_A
\end{equation*}
where $\mathds{P}_{X}$ denotes the set of Pauli strings over a region $X\subset [n]$. For ease of notations, we denote by $P_X$ the Pauli sub-string generated by $P\in\mathds{P}_n$ on a region $X\subset [n]$.

In this paper, we will assume that the underlying distribution $\mu$ has full support and has a conditional independence structure modeled by a (potentially unknown) hypergraph $\mathcal{H}=(V,H)$ with $n$ vertices corresponding to the qubits. In a nutshell, the hypergraph encodes the fact that, if two sets of vertices $A,C$ are shielded away from each other by a third set of vertices $B$, then the probability that a certain error occurs in vertices $A,C$ is independent conditioned on the errors that happened on $B$. To illustrate this concept, imagine three qubits arranged on a line such that qubit $1,2$ and $2,3$ interact but not $1,3$. Then we could assume that, once we fixed that a given error happened on qubit $2$, the probability of errors happening on $1$ and $3$ are independent.

The Hammersley-Clifford theorem~\cite{besag1974spatial} then implies that such a measure $\mu$ corresponds to a Gibbs distribution. This means that
\begin{equation}
\mu(P)\coloneqq \operatorname{exp}\Big(\sum_{\ell=1}^r\sum_{h=(i_1,\ldots,i_{\ell})\in H}\,\theta^{h}(P_{h})-C\Big)\,.\label{Gibbs}
\end{equation}
Here $\theta^{h}:\mathds{P}_{h}\to \mathbb{R}$ is a function that takes as inputs Pauli matrices on
qubits $i_1 \dots i_\ell$. These functions are referred to as
clique potentials. In the equation above, $C$ is a constant that ensures the distribution is normalized and is called the log-partition function. The parameter $r$ denotes the maximal range of a potential, i.e.~how many qubits it acts on. We also assume w.l.o.g. that $\sum_{Q\in\mathds{P}_h}\theta^h(Q)=0$, as we can always shift the potential by a constant without changing the distribution.

A hyperedge $h$ is called maximal if
no other hyperedge of strictly larger size contains $h$. A graph $G = (V, E)$ can be obtained by replacing
every hyperedge with a clique. In other words, each clique in the graph corresponds to a local interaction in $H$. Let $D$ be a bound on the maximum degree of the graph. We denote by $\Gamma(u)$ 
the neighbors of $u$. We further require the following conditions to ensure that we can learn the hypergraph:
\begin{itemize}
\item[(a)] Every edge $(i, j)\in E$ is contained in some hyperedge $h \in H$ where the corresponding tensor is non-zero.
\item[(b)] For every maximal hyperedge $h \in H$, $\theta^h$ takes at least one value lower bounded by some constant $\alpha>0$ in absolute value.
\item[(c)] All the values $\theta^h$ takes are upper bounded by a constant $\beta>0$ in absolute value.
\end{itemize}

\section{Learning the conditional independence structure}

Our first task is to learn the set $H$ of hyperedges. For this, we adapt a greedy approach by Bresler \cite{bresler2015efficiently}. Given a site $u$ and regions $I,S$, first introduce the functional
\begin{align}
\nu_{u,I|S} \coloneqq \mathbb{E}_{R,G}\mathbb{E}_{P_S\sim \mu_S}\Big[|\Delta_{u,I|S}|\Big]\,,
\end{align}
where $\Delta_{u,I|S}\coloneqq \mathbb{P}_\mu(P_u=R,P_I=G|P_S)-\mathbb{P}_{\mu}(P_u=R|P_S)\mathbb{P}_{\mu}(P_I=G|P_S)$, $R$ is a configuration drawn uniformly at random from $\mathds{P}_u$ and $G$ is an $|I|$-tuple drawn uniformly at random from $\mathds{P}_I$. Finally, $P_S$ is drawn according to the unknown distribution $\mu_S$. This quantity will play a crucial role in our algorithm. In words, it quantifies the amount by which the site $u$ and region $I$ are correlated, conditioned on the nature of the configuration of Pauli matrices in the region $S$. In spirit, the larger $\nu_{u,I|S}$, the more likely $u$ and $I$ are to be neighboring, and it will be $0$ if $S$ shields $u$ away from $I$. Since we do not have direct access to this quantity due to its dependence on the unknown distribution $\mu_{\{u\}\cup I\cup S}$, we also define an estimate for it in the standard way:
\begin{align}
\widehat{\nu}_{u,I|S}\coloneqq \mathbb{E}_{R,G}\,\widehat{\mathbb{E}}_{P_S}\,\Big[|\widehat{\Delta}_{u,I|S}|\Big]\,.
\end{align}
Here, we assume access to an approximation $\hat{\mu}_{\{u\}\cup I\cup S}$ of the distribution $\mu_{\{u\}\cup I\cup S}$ which we will later construct from access to the Fourier coefficients introduced in \eqref{eq:Fourier}. The average $\widehat{\mathbb{E}}_{P_S}$ is therefore defined with respect to the marginal $\hat{\mu}_S$, whereas $\widehat{\Delta}_{u,I|S}\coloneqq \mathbb{P}_{\hat{\mu}_{\{u\}\cup I\cup S}}(P_u=R,P_I=G|P_S)-\mathbb{P}_{\hat{\mu}_{\{u\}\cup I\cup S}}(P_u=R|P_S)\mathbb{P}_{\hat{\mu}_{\{u\}\cup I\cup S}}(P_I=G|P_S)$. The following algorithm will succeed in finding the neighborhood of a node $u$ with high probability as long as $\widehat{\nu}_{u,I|S}$ is sufficiently close to $\nu_{u,I|S}$. The algorithm also depends on two parameters $\tau,L>0$, where $\tau$ serves as a thresholding constant, whereas $L$ upper bounds the size of the superset of a neighborhood of $u$. Both these parameters will be fixed later on. 

\begin{algorithm}[H]
  \caption{NeighborhoodLearning}\label{protocollearning2}
  \label{protocol}
  \KwIn{Vertex $u\in [n]$, $L,\tau\ge 0$.}
  \KwOut{Estimated neighborhood $S$ of $u$.}
  \BlankLine
  Initialize $S\coloneqq \emptyset$\;
  \While{$|S|\le L$ and $\exists I\subset [n]\backslash S$, $|I|\le r-1$ s.t.~$\widehat{\nu}_{u,I|S}>\tau$ }{$S\leftarrow  S\cup I$\;}
    \For{$i\in S$}{if $\widehat{\nu}_{u,i|S\backslash i}<\tau$, $S\leftarrow S\backslash \{i\}$;}
  \Return{$S$.}
  \end{algorithm}

Next, given $\ell\in\mathbb{N}$ and $\epsilon>0$, we denote by $A(\ell,\epsilon)$ the event that for all $u,I,S$ with $|I|\le r-1$ and $|S|\le \ell$ simultaneously,
$|\nu_{u,I|S}-\widehat{\nu}_{u,I|S}|\le \epsilon$. In particular, we denote $A\equiv A(L,\tau/2)$. We further choose $\tau\coloneqq \frac{2\alpha^2\delta^{r+d-1}}{r^{2r}4^{r+1}\binom{D}{r-1}\gamma e^{2\gamma}}$ and $L\coloneqq (8/\tau^2)\ln (4)$, with $\gamma\coloneqq \sup_{u\in [n]}\sum_{\ell=1}^r\sum_{i_2<\dots <i_\ell}|\theta^{ui_2\dots i_\ell}|_{\operatorname{max}}$ and $\eta\coloneqq \frac{1}{4}e^{-2\gamma}$. Here $|T|_{\operatorname{max}}$ denotes the maximum entry of a tensor $T$. In \cite[Theorem 5.7]{hamilton2017information}, it is shown that whenever the event $A$ occurs, Algorithm \ref{protocollearning2} returns the correct set of neighbors of $u$ for every node $u$. It remains to estimate the number of samples required for event $A$ to occur with high probability. This is done in the next Lemma.
\begin{lemma}\label{lemmasamplecomplexity}
The event $A$ occurs with probability $1-\delta$ when given access to 
\begin{align}\label{equ:number_copies_learn_structure}
\mathcal{O}\left(\log\left(\frac{1}{1-p_0}\right)\,\frac{4^{3r+5L}}{\big(C_{\operatorname{SPAM}}\,\tau\,\eta^L\big)^2 }\log\frac{n^{r+L}}{\delta}\right)
\end{align}
copies of the channel $\cP$, where $p_0=\mu(\mathds{I})$.
\end{lemma}

\begin{proof}
Since we assume that $|I|\le r-1$ and $|S|\le L$, the set $\{u\}\cup I\cup S$ has cardinality at most $r+L $. The number of possible subsets of $[n]$ with at most $r+L$ elements is upper bounded by $n^{r+L}$, and each such region corresponds to at most $4^{r+L}$ different configurations. Therefore, to learn all the Fourier coefficients $\alpha_P$ of Pauli strings of weight at most $r+L$ to precision $\tilde{\epsilon}$ with probability $1-\delta$, it suffices to run the protocol of Cor. \ref{ref:corSPAMlearning}, which requires access to $\mathcal{O}(-\log(1-p_0)\,3^{r+L}\tilde{\epsilon}^{-2}C_{\operatorname{SPAM}}^{-2}\log((r+L)n^{r+L}4^{r+L}\delta^{-1}))$ copies of the channel $\cP$. Once the coefficients $\alpha_P$ are all known, we compute their inverse Fourier transform, since for any $P_A$ supported in $A\subset [n]$,
\begin{align}\label{marginalsfromfourier}
\mu_A(P_A)=\sum_{Q\in \mathbf{P}_A}\,(-1)^{P_A.Q}\,\alpha_Q\,.
\end{align}
Now, knowing the coefficients $\alpha_Q$ to precision $\tilde{\epsilon}$ implies that we can compute the coefficients $\mu_A(P_A)$ to precision at least $4^{r+L}\tilde{\epsilon}$. By \cite[Lemma 5.2]{hamilton2017information}, this further implies that we can compute $\nu_{u,I|S}$ to precision $\epsilon$ as long as $4^{r+L}\tilde{\epsilon}\le \epsilon 4^{-L}\frac{\eta^L}{5}$. The result follows from choosing $\tilde{\epsilon}= \frac{\tau}{2} 4^{-2 L -r}\frac{\eta^L}{5}$.
\end{proof}

Combining this with \cite[Theorem 5.7]{hamilton2017information}, we conclude the following.
\begin{thm}\label{learningstructure}
There exists a protocol robust to $\operatorname{SPAM}$ errors which uses a number of queries to $\cP$ as in Eq.~\eqref{equ:number_copies_learn_structure} and returns the conditional independence structure $G$ with probability $1-\delta$. Each query only requires implementing $1-$qubit Clifford gates circuits of the form of that of Figure \ref{fig:robust_parallelization_noisy} and at most $\mathcal{O}(\log(1/(1-p_0))$ sequential queries to $\cP$ are made per run of a circuit. Moreover, the protocol queries $n$ times the protocol of Algorithm \ref{protocollearning2}, whose run time is of order $\widetilde{\mathcal{O}}(Ln^r)$.
\end{thm}

\section{Learning the coefficients of the interaction}

Next, we assume that we have already successfully learned the underlying conditional independence structure, for instance by running Algorithm \ref{protocollearning2}, and we aim at learning the values of the coefficients $\theta^{h}(P_{\{h\}})$ for any $h\in\mathcal{H}$. Since the distribution $\mu$ is assumed to be Gibbs, it is well-known that this task can be achieved efficiently with access to $\mathcal{O}(\log(n))$ samples drawn from local marginals. The idea behind this fact is that the effective Hamiltonian of the marginal of a Gibbs measure is also local. Thus, learning
can be performed locally (see e.g.~\cite{note_anurag} for a proof in the more general context of a commuting Hamiltonian). The following result therefore follows from the fact that these marginals can be accessed via the Fourier coefficients, cf.~Equation \ref{marginalsfromfourier}. In what follows, for any hyperedge $h\in\mathcal{H}$, we denote by $R_{h}\subseteq[n]$ the smallest set of vertices that contains $h$ in its strict interior (i.e.~such that $h$ does not overlap with the inner boundary $\partial_-R_h$ of $R_h$, that is vertices which are contained in other hyperedges $h'\ne h$). 
 We will also need the orthogonal basis $\{\chi_P\coloneqq Q\mapsto (-1)^{P.Q}\}_{P\in\mathds{P}_1}$ of characters over $\mathds{P}_1\simeq  \mathbb{Z}_2\times \mathbb{Z}_2$, and more generally define the product basis $\{\chi_{P}\}_{P\in\mathbb{P}_A}$ over a region $A$ by taking products of the single qubit characters.

\begin{algorithm}[H]
  \caption{CoefficientLearning}\label{protocollearning3}
  \KwIn{$h\in H$.}
  \KwOut{Estimated interactions $\widehat{\theta}^h$.}
  \BlankLine
  \Return{$\widehat{\theta}^h\coloneqq \frac{1}{4^{|R_h|}}\sum_{P\in \mathbb{P}_{h}} \chi_P\, \langle \chi_P,\,\log(\widehat{\mu}_{R_h})\rangle $.}
  \end{algorithm}

\begin{thm}\label{learncoeffs}
Assuming that the hypergraph $\mathcal{H}=(V,H)$ is known, there exists a protocol robust to $\operatorname{SPAM}$ errors which uses a number of queries to $\cP$ of order
\begin{equation}
\mathcal{O}\left(\log\left(\frac{1}{1-p_0}\right)\,\frac{e^{\beta rD^{r+1}}4^{6rD^{r+1}}}{(\epsilon\,C_{\operatorname{SPAM}})^{2}}\log\frac{n^{rD^{r+1}}}{\delta}\right)\,.
\end{equation}
 It utilizes at most $\mathcal{O}(\log(1/(1-p_0)))$ sequential queries to $\cP$ and circuits consisting of single-qubit Clifford gates like in Figure \ref{fig:robust_parallelization_noisy}, and returns, for each $h\in H$, an estimator $\hat{\theta}^h$ of $\theta^h$ such that $\|\widehat{\theta}^h-\theta^h\|_\infty\le \epsilon$ with probability $1-\delta$. 
Furthermore, the algorithm queries $\mathcal{O}(n^{rD^{r+1}})$ times the protocol of Algorithm \ref{protocollearning3}, whose run time is of order $\widetilde{\mathcal{O}}(1)$.

\end{thm}

\begin{proof}
By a simple adaptation of the proof in \cite{note_anurag}, it can be shown that the estimator $\widehat{\theta}^h$ from Algorithm \ref{protocollearning3} satisfies $\|\widehat{\theta}^h-\theta^h\|_\infty\le \epsilon$ as soon as $\big\|\mu_{R_h}-\widehat{\mu}_{R_h}\big\|_{\operatorname{TV}}\le {\epsilon}\,e^{-|R_h|\,|\theta^h|}4^{-|R_h|-2r}$. Since $|R_h|\le rD^{r+1}$ and $|\theta^h|\le \beta$ by assumption, the task then reduces to that of estimating all the marginals $\mu_{R_h}$ up to an error of ${\epsilon}e^{-\beta\,rD^{r+1}}4^{-rD^{r+1}-2r}$ in total variation. This can be done in a SPAM error robust manner as in the proof of Lemma \ref{lemmasamplecomplexity}: since the number of possible subsets of $[n]$ with at most $rD^{r+1}$ elements is upper bounded by $rD^{r+1}n^{rD^{r+1}}$, and since each such region corresponds to at most $4^{rD^{r+1}}$ different configurations, to learn all the corresponding Fourier coefficients to precision $\epsilon'$ with probability $1-\delta$, it suffices to run the protocol of Corollary \ref{ref:corSPAMlearning}, which requires access to $\mathcal{O}(-\log(1-p_0)\,3^{rD^{r+1}}{\epsilon'}^{-2}C_{\operatorname{SPAM}}^{-2}\log(rD^{r+1}n^{rD^{r+1}}4^{rD^{r+1}}\delta^{-1}))$ copies of the channel $\cP$. Moreover, by Equation \ref{marginalsfromfourier}, knowing these Fourier coefficients to precision $\epsilon'$ implies that we can approximate the marginals in sup norm to precision $4^{rD^{r+1}}\epsilon'$, and therefore in total variation to precision $4^{2rD^{r+1}}\epsilon'$. The result follows after setting $4^{2rD^{r+1}}\epsilon'={\epsilon}e^{-\beta\,rD^{r+1}}4^{-rD^{r+1}-2r}$.
\end{proof}
Note that the fact that we can estimate multiple Fourier coefficients in parallel is crucial to ensure the $\log(n)$ scaling of the sample complexity, something previous approaches to learning the coefficients~\cite{flammia2021averaged,harper2020efficient,flammia2021pauli,helsen2022general} did not achieve.

At this stage, we turn our attention to our original problem of learning local Pauli channels in diamond norm. For this, we introduce some few more notations: Given $Q_h\in \mathds{P}_h$, we denote $e_{Q_h}\coloneqq \frac{1}{2^{n}}\,\langle \mu,\,\chi_{Q_h}\rangle$ and $\lambda_{Q_h}\coloneqq \frac{1}{2^n}\,\langle \chi_{Q_h},\,\log\mu\rangle$, similarly, we denote by $\widehat{e}_{Q_h}\coloneqq \frac{1}{2^{n}}\,\langle \widehat{\mu},\,\chi_{Q_h}\rangle$ and $\lambda_{Q_h}\coloneqq \frac{1}{2^n}\,\langle \chi_{Q_h},\,\log\widehat{\mu}\rangle$, where the inner products are with respect to the $\ell_2$ norm on $\mathds{P}_n$. By Plancherel identity, it is clear that $\|\lambda-\widehat{\lambda}\|_{\ell_2}^2=\sum_{|h|\le r}\sum_{P_h\in\mathds{P}_h}|\theta^h(P_h)-\widehat{\theta}^h(P_h)|^2$. Moreover, for classical Gibbs measures, the following strong convexity property holds (see e.g. \cite{montanari2015computational,vuffray2016interaction}):
\begin{align*}
\|e-\widehat{e}\|_{\ell_2}\le \Gamma\,\|\lambda-\widehat{\lambda}\|_{\ell_2}\,,
\end{align*}
for some constant $\Gamma=\mathcal{O}(1)$ that depends exponentially on the degree $D$ but is independent of $n$. Here the euclidean space $\ell_2$ is with respect to the set of parameters $Q_h$ for $h\in H$, i.e.
\begin{align*}
\|x\|_{\ell_2}^2\coloneqq \sum_{h\in H}\sum_{Q_h\in\mathds{P}_h}\,|x_{Q_h}|^2\,.
\end{align*}

\begin{cor}
The learning schemes described in Theorems \ref{learningstructure} and \ref{learncoeffs} with re-scaled required precision $\mathcal{O}(\epsilon n^{-r})$ provides an estimator $\hat{\mu}$ of the unknown measure $\mu$ whose corresponding Pauli channel $\widehat{\cP}$, defined by $\widehat{\cP}(\rho)\coloneqq \sum_{P\in\mathds{P}_n}\widehat{\mu}(P)\,P\rho P$, satisfies
\begin{align*}
\|\cP-\widehat{\cP}\|_\diamond\le \epsilon\,.
\end{align*}
In particular, this strategy provides a classical description of the channel $\widehat{\cP}$ with access to $\widetilde{\mathcal{O}}(n^{r}\epsilon^{-2})$ uses of the channel $\cP$ and a run-time of $\widetilde{\mathcal{O}}(n^{r(2+D^{r+1})}\epsilon^{-2})$.
\end{cor}

\begin{proof}
We first recall that \cite{magesan2012characterizing}:
\begin{align*}
\|\cP-\widehat{P}\|^2_\diamond&= (2\,\|\mu-\widehat{\mu}\|_{\operatorname{TV}})^2\\
&\le D(\mu\|\widehat{\mu})+D(\widehat{\mu}\|\mu)\,,
\end{align*}
where in the second line, we use Pinsker's inequality which relates the total variation distance between two probability measures to their relative entropies:
\begin{align*}
D(\mu_1\|\mu_2)\coloneqq \sum_{P\in\mathds{P}_n}\,\mu_1(P)\,\big(\log(\mu_1(P))-\log(\mu_2(P))\big)\,.
\end{align*}
By using the definitions for $\widehat{\mu}$ and $\mu$, we further can show that
\begin{align*}
D(\mu\|\widehat{\mu})+D(\widehat{\mu}\|\mu)&=\sum_{|h|\le r}\sum_{Q_h\in\mathds{P}_h} ( e_{Q_h}-\widehat{e}_{Q_h})(\lambda_{Q_h}-\widehat{\lambda}_{Q_h})\\
&\le \|e-\widehat{e}\|_{\ell_2}\,\|\lambda-\widehat{\lambda}\|_{\ell_2}\\
&\le \Gamma\,\|\lambda-\widehat{\lambda}\|_{\ell^2}^2\\
&=\Gamma\,\sum_{|h|\le r}\sum_{P_h\in\mathds{P}_h}\,|\theta^h(P_h)-\widehat{\theta}^h(P_h)|^2\\
&= \mathcal{O}(n^r)\,\max_{h}\|\theta^h-\widehat{\theta}^h\|_\infty^2\,.
\end{align*}
The result follows after re-scaling $\epsilon$ to $\epsilon n^{-r}$ in Theorems \ref{learningstructure} and \ref{learncoeffs} in order to get the required precision.
\end{proof}

{\it Acknowledgements.} CR acknowledges financial support from the ANR project QTraj (ANR-20-CE40-0024-01) of the French National Research Agency (ANR), as well as from the Humboldt Foundation. This work is part of HQI initiative (\href{www.hqi.fr}{www.hqi.fr}) and is supported by France 2030 under the French National Research Agency award number “ANR-22-PNCQ-0002”.
We thank Matthias Caro, Steve Flammia and Aadil Oufkir for fruitful discussions and comments on the draft.

\bibliographystyle{apsrev4-2}
\bibliography{references}

\begin{thebibliography}{26}%
\makeatletter
\providecommand \@ifxundefined [1]{%
 \@ifx{#1\undefined}
}%
\providecommand \@ifnum [1]{%
 \ifnum #1\expandafter \@firstoftwo
 \else \expandafter \@secondoftwo
 \fi
}%
\providecommand \@ifx [1]{%
 \ifx #1\expandafter \@firstoftwo
 \else \expandafter \@secondoftwo
 \fi
}%
\providecommand \natexlab [1]{#1}%
\providecommand \enquote  [1]{``#1''}%
\providecommand \bibnamefont  [1]{#1}%
\providecommand \bibfnamefont [1]{#1}%
\providecommand \citenamefont [1]{#1}%
\providecommand \href@noop [0]{\@secondoftwo}%
\providecommand \href [0]{\begingroup \@sanitize@url \@href}%
\providecommand \@href[1]{\@@startlink{#1}\@@href}%
\providecommand \@@href[1]{\endgroup#1\@@endlink}%
\providecommand \@sanitize@url [0]{\catcode `\\12\catcode `\$12\catcode
  `\&12\catcode `\#12\catcode `\^12\catcode `\_12\catcode `\%12\relax}%
\providecommand \@@startlink[1]{}%
\providecommand \@@endlink[0]{}%
\providecommand \url  [0]{\begingroup\@sanitize@url \@url }%
\providecommand \@url [1]{\endgroup\@href {#1}{\urlprefix }}%
\providecommand \urlprefix  [0]{URL }%
\providecommand \Eprint [0]{\href }%
\providecommand \doibase [0]{https://doi.org/}%
\providecommand \selectlanguage [0]{\@gobble}%
\providecommand \bibinfo  [0]{\@secondoftwo}%
\providecommand \bibfield  [0]{\@secondoftwo}%
\providecommand \translation [1]{[#1]}%
\providecommand \BibitemOpen [0]{}%
\providecommand \bibitemStop [0]{}%
\providecommand \bibitemNoStop [0]{.\EOS\space}%
\providecommand \EOS [0]{\spacefactor3000\relax}%
\providecommand \BibitemShut  [1]{\csname bibitem#1\endcsname}%
\let\auto@bib@innerbib\@empty
\bibitem [{\citenamefont {Fawzi}\ \emph {et~al.}(2023)\citenamefont {Fawzi},
  \citenamefont {Oufkir},\ and\ \citenamefont {França}}]{fawzi2023lower}%
  \BibitemOpen
  \bibfield  {author} {\bibinfo {author} {\bibfnamefont {O.}~\bibnamefont
  {Fawzi}}, \bibinfo {author} {\bibfnamefont {A.}~\bibnamefont {Oufkir}},\ and\
  \bibinfo {author} {\bibfnamefont {D.~S.}\ \bibnamefont {França}},\
  }\href@noop {} {\bibinfo {title} {Lower bounds on learning pauli channels}}
  (\bibinfo {year} {2023}),\ \Eprint {https://arxiv.org/abs/2301.09192}
  {arXiv:2301.09192 [quant-ph]} \BibitemShut {NoStop}%
\bibitem [{\citenamefont {Flammia}\ and\ \citenamefont
  {Wallman}(2020)}]{flammia2020efficient}%
  \BibitemOpen
  \bibfield  {author} {\bibinfo {author} {\bibfnamefont {S.~T.}\ \bibnamefont
  {Flammia}}\ and\ \bibinfo {author} {\bibfnamefont {J.~J.}\ \bibnamefont
  {Wallman}},\ }\href@noop {} {\bibfield  {journal} {\bibinfo  {journal} {ACM
  Transactions on Quantum Computing}\ }\textbf {\bibinfo {volume} {1}},\
  \bibinfo {pages} {1} (\bibinfo {year} {2020})}\BibitemShut {NoStop}%
\bibitem [{\citenamefont {Chen}\ \emph {et~al.}(2022)\citenamefont {Chen},
  \citenamefont {Zhou}, \citenamefont {Seif},\ and\ \citenamefont
  {Jiang}}]{chen2022quantum}%
  \BibitemOpen
  \bibfield  {author} {\bibinfo {author} {\bibfnamefont {S.}~\bibnamefont
  {Chen}}, \bibinfo {author} {\bibfnamefont {S.}~\bibnamefont {Zhou}}, \bibinfo
  {author} {\bibfnamefont {A.}~\bibnamefont {Seif}},\ and\ \bibinfo {author}
  {\bibfnamefont {L.}~\bibnamefont {Jiang}},\ }\href@noop {} {\bibfield
  {journal} {\bibinfo  {journal} {Physical Review A}\ }\textbf {\bibinfo
  {volume} {105}},\ \bibinfo {pages} {032435} (\bibinfo {year}
  {2022})}\BibitemShut {NoStop}%
\bibitem [{\citenamefont {Flammia}\ and\ \citenamefont
  {O'Donnell}(2021)}]{flammia2021pauli}%
  \BibitemOpen
  \bibfield  {author} {\bibinfo {author} {\bibfnamefont {S.~T.}\ \bibnamefont
  {Flammia}}\ and\ \bibinfo {author} {\bibfnamefont {R.}~\bibnamefont
  {O'Donnell}},\ }\href@noop {} {\bibfield  {journal} {\bibinfo  {journal}
  {Quantum}\ }\textbf {\bibinfo {volume} {5}},\ \bibinfo {pages} {549}
  (\bibinfo {year} {2021})}\BibitemShut {NoStop}%
\bibitem [{\citenamefont {Wagner}\ \emph {et~al.}(2023)\citenamefont {Wagner},
  \citenamefont {Kampermann}, \citenamefont {Bru\ss{}},\ and\ \citenamefont
  {Kliesch}}]{PhysRevLett.130.200601}%
  \BibitemOpen
  \bibfield  {author} {\bibinfo {author} {\bibfnamefont {T.}~\bibnamefont
  {Wagner}}, \bibinfo {author} {\bibfnamefont {H.}~\bibnamefont {Kampermann}},
  \bibinfo {author} {\bibfnamefont {D.}~\bibnamefont {Bru\ss{}}},\ and\
  \bibinfo {author} {\bibfnamefont {M.}~\bibnamefont {Kliesch}},\ }\href
  {https://doi.org/10.1103/PhysRevLett.130.200601} {\bibfield  {journal}
  {\bibinfo  {journal} {Phys. Rev. Lett.}\ }\textbf {\bibinfo {volume} {130}},\
  \bibinfo {pages} {200601} (\bibinfo {year} {2023})}\BibitemShut {NoStop}%
\bibitem [{\citenamefont {Wallman}\ and\ \citenamefont
  {Emerson}(2016)}]{PhysRevA.94.052325}%
  \BibitemOpen
  \bibfield  {author} {\bibinfo {author} {\bibfnamefont {J.~J.}\ \bibnamefont
  {Wallman}}\ and\ \bibinfo {author} {\bibfnamefont {J.}~\bibnamefont
  {Emerson}},\ }\href {https://doi.org/10.1103/PhysRevA.94.052325} {\bibfield
  {journal} {\bibinfo  {journal} {Phys. Rev. A}\ }\textbf {\bibinfo {volume}
  {94}},\ \bibinfo {pages} {052325} (\bibinfo {year} {2016})}\BibitemShut
  {NoStop}%
\bibitem [{\citenamefont {Harper}\ \emph {et~al.}(2020)\citenamefont {Harper},
  \citenamefont {Flammia},\ and\ \citenamefont
  {Wallman}}]{harper2020efficient}%
  \BibitemOpen
  \bibfield  {author} {\bibinfo {author} {\bibfnamefont {R.}~\bibnamefont
  {Harper}}, \bibinfo {author} {\bibfnamefont {S.~T.}\ \bibnamefont
  {Flammia}},\ and\ \bibinfo {author} {\bibfnamefont {J.~J.}\ \bibnamefont
  {Wallman}},\ }\href@noop {} {\bibfield  {journal} {\bibinfo  {journal}
  {Nature Physics}\ }\textbf {\bibinfo {volume} {16}},\ \bibinfo {pages} {1184}
  (\bibinfo {year} {2020})}\BibitemShut {NoStop}%
\bibitem [{\citenamefont {Harper}\ and\ \citenamefont
  {Flammia}(2023)}]{harper2023learning}%
  \BibitemOpen
  \bibfield  {author} {\bibinfo {author} {\bibfnamefont {R.}~\bibnamefont
  {Harper}}\ and\ \bibinfo {author} {\bibfnamefont {S.~T.}\ \bibnamefont
  {Flammia}},\ }\href@noop {} {\bibinfo {title} {Learning correlated noise in a
  39-qubit quantum processor}} (\bibinfo {year} {2023}),\ \Eprint
  {https://arxiv.org/abs/2303.00780} {arXiv:2303.00780 [quant-ph]} \BibitemShut
  {NoStop}%
\bibitem [{\citenamefont {Bresler}(2015)}]{bresler2015efficiently}%
  \BibitemOpen
  \bibfield  {author} {\bibinfo {author} {\bibfnamefont {G.}~\bibnamefont
  {Bresler}},\ }in\ \href@noop {} {\emph {\bibinfo {booktitle} {Proceedings of
  the forty-seventh annual ACM symposium on Theory of computing}}}\ (\bibinfo
  {year} {2015})\ pp.\ \bibinfo {pages} {771--782}\BibitemShut {NoStop}%
\bibitem [{\citenamefont {Fran{\c{c}}a}\ and\ \citenamefont
  {Hashagen}(2018)}]{francca2018approximate}%
  \BibitemOpen
  \bibfield  {author} {\bibinfo {author} {\bibfnamefont {D.~S.}\ \bibnamefont
  {Fran{\c{c}}a}}\ and\ \bibinfo {author} {\bibfnamefont {A.}~\bibnamefont
  {Hashagen}},\ }\href@noop {} {\bibfield  {journal} {\bibinfo  {journal}
  {Journal of Physics A: Mathematical and Theoretical}\ }\textbf {\bibinfo
  {volume} {51}},\ \bibinfo {pages} {395302} (\bibinfo {year}
  {2018})}\BibitemShut {NoStop}%
\bibitem [{\citenamefont {Helsen}\ \emph {et~al.}(2022)\citenamefont {Helsen},
  \citenamefont {Roth}, \citenamefont {Onorati}, \citenamefont {Werner},\ and\
  \citenamefont {Eisert}}]{helsen2022general}%
  \BibitemOpen
  \bibfield  {author} {\bibinfo {author} {\bibfnamefont {J.}~\bibnamefont
  {Helsen}}, \bibinfo {author} {\bibfnamefont {I.}~\bibnamefont {Roth}},
  \bibinfo {author} {\bibfnamefont {E.}~\bibnamefont {Onorati}}, \bibinfo
  {author} {\bibfnamefont {A.~H.}\ \bibnamefont {Werner}},\ and\ \bibinfo
  {author} {\bibfnamefont {J.}~\bibnamefont {Eisert}},\ }\href@noop {}
  {\bibfield  {journal} {\bibinfo  {journal} {PRX Quantum}\ }\textbf {\bibinfo
  {volume} {3}},\ \bibinfo {pages} {020357} (\bibinfo {year}
  {2022})}\BibitemShut {NoStop}%
\bibitem [{\citenamefont {Helsen}\ \emph {et~al.}(2019)\citenamefont {Helsen},
  \citenamefont {Xue}, \citenamefont {Vandersypen},\ and\ \citenamefont
  {Wehner}}]{helsen2019new}%
  \BibitemOpen
  \bibfield  {author} {\bibinfo {author} {\bibfnamefont {J.}~\bibnamefont
  {Helsen}}, \bibinfo {author} {\bibfnamefont {X.}~\bibnamefont {Xue}},
  \bibinfo {author} {\bibfnamefont {L.~M.}\ \bibnamefont {Vandersypen}},\ and\
  \bibinfo {author} {\bibfnamefont {S.}~\bibnamefont {Wehner}},\ }\href@noop {}
  {\bibfield  {journal} {\bibinfo  {journal} {npj Quantum Information}\
  }\textbf {\bibinfo {volume} {5}},\ \bibinfo {pages} {71} (\bibinfo {year}
  {2019})}\BibitemShut {NoStop}%
\bibitem [{\citenamefont {Magesan}\ \emph {et~al.}(2012)\citenamefont
  {Magesan}, \citenamefont {Gambetta},\ and\ \citenamefont
  {Emerson}}]{magesan2012characterizing}%
  \BibitemOpen
  \bibfield  {author} {\bibinfo {author} {\bibfnamefont {E.}~\bibnamefont
  {Magesan}}, \bibinfo {author} {\bibfnamefont {J.~M.}\ \bibnamefont
  {Gambetta}},\ and\ \bibinfo {author} {\bibfnamefont {J.}~\bibnamefont
  {Emerson}},\ }\href@noop {} {\bibfield  {journal} {\bibinfo  {journal}
  {Physical Review A}\ }\textbf {\bibinfo {volume} {85}},\ \bibinfo {pages}
  {042311} (\bibinfo {year} {2012})}\BibitemShut {NoStop}%
\bibitem [{\citenamefont {Heinrich}\ \emph {et~al.}(2023)\citenamefont
  {Heinrich}, \citenamefont {Kliesch},\ and\ \citenamefont
  {Roth}}]{heinrich2023randomized}%
  \BibitemOpen
  \bibfield  {author} {\bibinfo {author} {\bibfnamefont {M.}~\bibnamefont
  {Heinrich}}, \bibinfo {author} {\bibfnamefont {M.}~\bibnamefont {Kliesch}},\
  and\ \bibinfo {author} {\bibfnamefont {I.}~\bibnamefont {Roth}},\ }\href@noop
  {} {\bibinfo {title} {Randomized benchmarking with random quantum circuits}}
  (\bibinfo {year} {2023}),\ \Eprint {https://arxiv.org/abs/2212.06181}
  {arXiv:2212.06181 [quant-ph]} \BibitemShut {NoStop}%
\bibitem [{\citenamefont {França}\ \emph {et~al.}(2022)\citenamefont
  {França}, \citenamefont {Markovich}, \citenamefont {Dobrovitski},
  \citenamefont {Werner},\ and\ \citenamefont
  {Borregaard}}]{franca2022efficient}%
  \BibitemOpen
  \bibfield  {author} {\bibinfo {author} {\bibfnamefont {D.~S.}\ \bibnamefont
  {França}}, \bibinfo {author} {\bibfnamefont {L.~A.}\ \bibnamefont
  {Markovich}}, \bibinfo {author} {\bibfnamefont {V.~V.}\ \bibnamefont
  {Dobrovitski}}, \bibinfo {author} {\bibfnamefont {A.~H.}\ \bibnamefont
  {Werner}},\ and\ \bibinfo {author} {\bibfnamefont {J.}~\bibnamefont
  {Borregaard}},\ }\href@noop {} {\bibinfo {title} {Efficient and robust
  estimation of many-qubit hamiltonians}} (\bibinfo {year} {2022}),\ \Eprint
  {https://arxiv.org/abs/2205.09567} {arXiv:2205.09567 [quant-ph]} \BibitemShut
  {NoStop}%
\bibitem [{\citenamefont {Kunjummen}\ \emph {et~al.}(2023)\citenamefont
  {Kunjummen}, \citenamefont {Tran}, \citenamefont {Carney},\ and\
  \citenamefont {Taylor}}]{kunjummen2023shadow}%
  \BibitemOpen
  \bibfield  {author} {\bibinfo {author} {\bibfnamefont {J.}~\bibnamefont
  {Kunjummen}}, \bibinfo {author} {\bibfnamefont {M.~C.}\ \bibnamefont {Tran}},
  \bibinfo {author} {\bibfnamefont {D.}~\bibnamefont {Carney}},\ and\ \bibinfo
  {author} {\bibfnamefont {J.~M.}\ \bibnamefont {Taylor}},\ }\href@noop {}
  {\bibfield  {journal} {\bibinfo  {journal} {Physical Review A}\ }\textbf
  {\bibinfo {volume} {107}},\ \bibinfo {pages} {042403} (\bibinfo {year}
  {2023})}\BibitemShut {NoStop}%
\bibitem [{\citenamefont {Levy}\ \emph {et~al.}(2021)\citenamefont {Levy},
  \citenamefont {Luo},\ and\ \citenamefont {Clark}}]{levy2021classical}%
  \BibitemOpen
  \bibfield  {author} {\bibinfo {author} {\bibfnamefont {R.}~\bibnamefont
  {Levy}}, \bibinfo {author} {\bibfnamefont {D.}~\bibnamefont {Luo}},\ and\
  \bibinfo {author} {\bibfnamefont {B.~K.}\ \bibnamefont {Clark}},\ }\href@noop
  {} {\bibinfo {title} {Classical shadows for quantum process tomography on
  near-term quantum computers}} (\bibinfo {year} {2021}),\ \Eprint
  {https://arxiv.org/abs/2110.02965} {arXiv:2110.02965 [quant-ph]} \BibitemShut
  {NoStop}%
\bibitem [{\citenamefont {Watrous}(2018)}]{watrous2018theory}%
  \BibitemOpen
  \bibfield  {author} {\bibinfo {author} {\bibfnamefont {J.}~\bibnamefont
  {Watrous}},\ }\href@noop {} {\emph {\bibinfo {title} {The theory of quantum
  information}}}\ (\bibinfo  {publisher} {Cambridge university press},\
  \bibinfo {year} {2018})\BibitemShut {NoStop}%
\bibitem [{\citenamefont {Besag}(1974)}]{besag1974spatial}%
  \BibitemOpen
  \bibfield  {author} {\bibinfo {author} {\bibfnamefont {J.}~\bibnamefont
  {Besag}},\ }\href@noop {} {\bibfield  {journal} {\bibinfo  {journal} {Journal
  of the Royal Statistical Society: Series B (Methodological)}\ }\textbf
  {\bibinfo {volume} {36}},\ \bibinfo {pages} {192} (\bibinfo {year}
  {1974})}\BibitemShut {NoStop}%
\bibitem [{\citenamefont {Hamilton}\ \emph {et~al.}(2017)\citenamefont
  {Hamilton}, \citenamefont {Koehler},\ and\ \citenamefont
  {Moitra}}]{hamilton2017information}%
  \BibitemOpen
  \bibfield  {author} {\bibinfo {author} {\bibfnamefont {L.}~\bibnamefont
  {Hamilton}}, \bibinfo {author} {\bibfnamefont {F.}~\bibnamefont {Koehler}},\
  and\ \bibinfo {author} {\bibfnamefont {A.}~\bibnamefont {Moitra}},\
  }\href@noop {} {\bibfield  {journal} {\bibinfo  {journal} {Advances in Neural
  Information Processing Systems}\ }\textbf {\bibinfo {volume} {30}} (\bibinfo
  {year} {2017})}\BibitemShut {NoStop}%
\bibitem [{\citenamefont {Anshu}\ \emph {et~al.}()\citenamefont {Anshu},
  \citenamefont {Arunachalam}, \citenamefont {Kuwahara},\ and\ \citenamefont
  {Soleimanifar}}]{note_anurag}%
  \BibitemOpen
  \bibfield  {author} {\bibinfo {author} {\bibfnamefont {A.}~\bibnamefont
  {Anshu}}, \bibinfo {author} {\bibfnamefont {S.}~\bibnamefont {Arunachalam}},
  \bibinfo {author} {\bibfnamefont {T.}~\bibnamefont {Kuwahara}},\ and\
  \bibinfo {author} {\bibfnamefont {M.}~\bibnamefont {Soleimanifar}},\ }\href
  {https://people.eecs.berkeley.edu/~anuraganshu/Learning_commuting_hamiltonian.pdf}
  {\bibinfo {title} {Efficient learning of commuting {H}amiltonians on
  lattices}},\ \bibinfo {note} {unpublished notes avaible at Anurag Anshu's
  website,
  \href{https://people.eecs.berkeley.edu/~anuraganshu/Learning_commuting_hamiltonian.pdf}{(link
  to note)}}\BibitemShut {NoStop}%
\bibitem [{\citenamefont {Flammia}(2021)}]{flammia2021averaged}%
  \BibitemOpen
  \bibfield  {author} {\bibinfo {author} {\bibfnamefont {S.~T.}\ \bibnamefont
  {Flammia}},\ }\href@noop {} {\bibinfo {title} {Averaged circuit eigenvalue
  sampling}} (\bibinfo {year} {2021}),\ \Eprint
  {https://arxiv.org/abs/2108.05803} {arXiv:2108.05803 [quant-ph]} \BibitemShut
  {NoStop}%
\bibitem [{\citenamefont {Montanari}(2015)}]{montanari2015computational}%
  \BibitemOpen
  \bibfield  {author} {\bibinfo {author} {\bibfnamefont {A.}~\bibnamefont
  {Montanari}},\ }\bibfield  {journal} {\bibinfo  {journal} {Electronic Journal
  of Statistics}\ }\textbf {\bibinfo {volume} {9}},\ \href
  {https://doi.org/10.1214/15-ejs1059} {10.1214/15-ejs1059} (\bibinfo {year}
  {2015})\BibitemShut {NoStop}%
\bibitem [{\citenamefont {Vuffray}\ \emph {et~al.}(2016)\citenamefont
  {Vuffray}, \citenamefont {Misra}, \citenamefont {Lokhov},\ and\ \citenamefont
  {Chertkov}}]{vuffray2016interaction}%
  \BibitemOpen
  \bibfield  {author} {\bibinfo {author} {\bibfnamefont {M.}~\bibnamefont
  {Vuffray}}, \bibinfo {author} {\bibfnamefont {S.}~\bibnamefont {Misra}},
  \bibinfo {author} {\bibfnamefont {A.}~\bibnamefont {Lokhov}},\ and\ \bibinfo
  {author} {\bibfnamefont {M.}~\bibnamefont {Chertkov}},\ }\href@noop {}
  {\bibfield  {journal} {\bibinfo  {journal} {Advances in neural information
  processing systems}\ }\textbf {\bibinfo {volume} {29}} (\bibinfo {year}
  {2016})}\BibitemShut {NoStop}%
\bibitem [{\citenamefont {Devroye}\ \emph {et~al.}(2016)\citenamefont
  {Devroye}, \citenamefont {Lerasle}, \citenamefont {Lugosi},\ and\
  \citenamefont {Oliveira}}]{Devroye2016}%
  \BibitemOpen
  \bibfield  {author} {\bibinfo {author} {\bibfnamefont {L.}~\bibnamefont
  {Devroye}}, \bibinfo {author} {\bibfnamefont {M.}~\bibnamefont {Lerasle}},
  \bibinfo {author} {\bibfnamefont {G.}~\bibnamefont {Lugosi}},\ and\ \bibinfo
  {author} {\bibfnamefont {R.~I.}\ \bibnamefont {Oliveira}},\ }\bibfield
  {journal} {\bibinfo  {journal} {The Annals of Statistics}\ }\textbf {\bibinfo
  {volume} {44}},\ \href {https://doi.org/10.1214/16-aos1440}
  {10.1214/16-aos1440} (\bibinfo {year} {2016})\BibitemShut {NoStop}%
\bibitem [{\citenamefont {Harper}\ \emph {et~al.}(2019)\citenamefont {Harper},
  \citenamefont {Hincks}, \citenamefont {Ferrie}, \citenamefont {Flammia},\
  and\ \citenamefont {Wallman}}]{PhysRevA.99.052350}%
  \BibitemOpen
  \bibfield  {author} {\bibinfo {author} {\bibfnamefont {R.}~\bibnamefont
  {Harper}}, \bibinfo {author} {\bibfnamefont {I.}~\bibnamefont {Hincks}},
  \bibinfo {author} {\bibfnamefont {C.}~\bibnamefont {Ferrie}}, \bibinfo
  {author} {\bibfnamefont {S.~T.}\ \bibnamefont {Flammia}},\ and\ \bibinfo
  {author} {\bibfnamefont {J.~J.}\ \bibnamefont {Wallman}},\ }\href
  {https://doi.org/10.1103/PhysRevA.99.052350} {\bibfield  {journal} {\bibinfo
  {journal} {Phys. Rev. A}\ }\textbf {\bibinfo {volume} {99}},\ \bibinfo
  {pages} {052350} (\bibinfo {year} {2019})}\BibitemShut {NoStop}%
\end{thebibliography}%

\appendix
\section{SPAM robust parallel estimation of Pauli eigenvalues}
Our protocol requires us to estimate the eigenvalues of the Pauli channel for Paulis with small support, as these correspond to low-weight Fourier coefficients.
The estimation of eigenvalues of a Pauli channel has been considered in several works in the literature~\cite{flammia2020efficient,flammia2021pauli,flammia2021averaged}. In particular, using tools from randomized benchmarking techniques, it is possible to estimate these coefficients in a way that is robust against both state preparation and measurement errors (SPAM). 
This is a highly desirable feature for process tomography protocols since in current devices these errors are not negligible.
On the other hand, recent schemes inspired by the classical shadows framework for learning states~\cite{kunjummen2023shadow,levy2021classical,franca2022efficient} showed how to infer various parameters of a quantum channel from a single class of randomized experiments. 
Indeed, such schemes are able to estimate $m$ Pauli diagonal elements of the form
\begin{align}
2^{-n}\tr{P\cP(P)}
\end{align}
for Paulis of weight at most $w$ up to an error $\epsilon$ with probability of success at least $1-\delta$ from $\mathcal{O}(9^w\epsilon^{-2}\log(m\delta^{-1}))$ randomized experiments. Thus, we see that this protocol is highly parallelizable. However, it is not as robust to SPAM errors as randomized benchmarking. We will now show how to obtain the best of the two worlds.

We will denote by $\Phi^*$ the Hilbert Schmidt dual (i.e. evolution in Heisenberg picture) to the map $\Phi$, and by $\mathcal{Q}$ the action of unitary conjugation by $Q$, i.e.~$\mathcal{Q}(\rho)=Q\rho Q^\dagger$.

The class of circuits we consider are illustrated in Fig.~\ref{fig:robust_parallelization}. We first prepare the system in the $\ket{0^n}$ state, and proceed to apply one layer of one-qubit random Cliffords and one layer of random Pauli matrices. After that, we apply the target channel $k$ times. This is followed by another layer of random Paulis and $1-$qubit Cliffords, which are drawn independently from the previous ones.

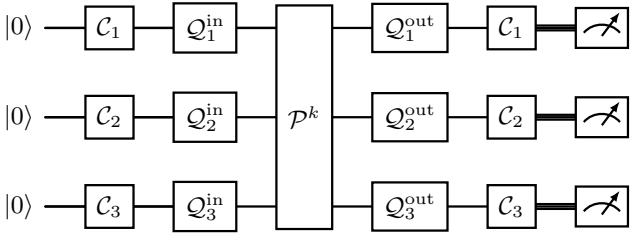
\begin{figure}
    \centering
    \resizebox{\columnwidth}{!}{

\begin{quantikz}
\lstick{$\ket{0}$} & \gate{\mathcal{C}_1} \qw  & \gate{\cQ_1^{\operatorname{in}}}\qw & \gate[3]{\cP^k} \qw & \gate{\cQ_1^{\operatorname{out}}}&\gate{\cC_1}& \meter   \cw\\
\lstick{$\ket{0}$} & \gate{\mathcal{C}_2} \qw  & \gate{\cQ_2^{\operatorname{in}}}\qw & & \gate{\cQ_2^{\operatorname{out}}}&\gate{\cC_2}& \meter   \cw\\
\lstick{$\ket{0}$} & \gate{\mathcal{C}_3} \qw  & \gate{\cQ_3^{\operatorname{in}}}\qw & & \gate{\cQ_3^{\operatorname{out}}}&\gate{\cC_3}& \meter   \cw
\end{quantikz}}
    \caption{Class of circuits we use for learning the Pauli eigenvalues on three qubits in the absence of SPAM. Here, $\mathcal{C}_i$ are independent uniformly distributed random one-qubit Clifford unitaries, $P_i,Q_i$ are uniform independent random Pauli matrices.}
    \label{fig:robust_parallelization}
\end{figure}

The post-processing of this data will be inspired by both the shadow process of~\cite{franca2022efficient} and the character randomized benchmarking of~\cite{francca2018approximate,helsen2019new}.
For a target Pauli string $P$, we define the function $\chi_P:\mathds{P}_n\to\{-1,1\}$ given by 
\begin{align}
    \chi_P(Q)=(-1)^{P\cdot Q}.
\end{align}
We also define the $P$-Twirling $\cT_P$ as 
\begin{align}\label{equ:twirling_Pauli}
\cT_P(X)\coloneqq \frac{1}{|\mathds{P}_n|}\sum_{Q\in \mathds{P}_n}\chi_P(Q)\,QXQ.
\end{align}
It follows from standard results of representation theory that $\cT_P(X)=2^{-n}\tr{PX}P$.
Furthermore, given a Pauli string $P$, an $n$ qubit Clifford unitary channel $\mathcal{C}=\otimes_{i=1}^n\mathcal{C}_i$ and bitstring $x\in\{0,1\}^n$ we define the function $\omega_P$ as 
\begin{align}
\omega_{P}(\cC,x)\coloneqq 3^{w(P)}\bra{x}C^\dagger PC\ket{x}\bra{0^n}C^\dagger PC\ket{0^n}\,,
\end{align}
Given a sample generated from the circuit as in Fig.~\ref{fig:robust_parallelization} where we used the channel $k$ times with the corresponding random Cliffords and Paulis and a reference Pauli string $P$, we consider the random variable 
\begin{align}\label{eq:RV}
\Omega_k\coloneqq \omega_{P}(\cC,X_k)\,\chi_{P}(Q^{\operatorname{in}})\,\chi_{P}(Q^{\operatorname{out}})\,,
\end{align}
where $X_k$ corresponds to the output to the circuit in Figure \ref{fig:robust_parallelization}. Let $\alpha^k_P\coloneqq 2^{-n}\tr{P\cP^k(P)}$ and $\alpha_P=\alpha_P^1$. We will now show that $\mathbb{E}[\Omega_k]=C\alpha_P^k$ for some constant $C$ that only depends on the amount of SPAM errors. Furthermore, we have that the variance of $\Omega_k$ is bounded as $\mathbb{V}(\Omega_k)\leq 3^{w(P)}$.
Thus, by using a median of means estimator, we can ensure that we obtain estimates of $C\alpha_P^k$ for different choices of $P$. By then fitting the curve $k\mapsto C\alpha_P^k$ for various values of $k$, as we will describe in more detail below, it is possible to obtain an estimate of $\alpha_P$ that is robust to SPAM errors.

Before we proceed to prove the claims made in the previous paragraph, let us first make some assumptions about how the noise acts. We will model the SPAM noise and the noise incurred by the random Paulis by assuming that between those two steps, there are channels $\Phi_1,\Phi_2$ acting on the systems. This is illustrated in Fig.~\ref{fig:robust_parallelization_noisy}. Note that we assume that the noise is independent of the gates and independent of $k$. The latter assumption is justified by the fact that the circuit we need to implement for the measurement and preparations is independent of $k$. The former is a standard assumption in RB, but recent works~\cite{helsen2022general} have relaxed it to weakly dependent noise. Generalizing it to this setting should be possible here as well, but would go beyond the scope of this work.
\begin{figure}
    \centering
    \resizebox{\columnwidth}{!}{
\begin{quantikz}
\lstick{$\ket{0}$} & \gate{\mathcal{C}_1} \qw &\gate[3]{\Phi_1}  & \gate{\cQ_1^{\operatorname{in}}}\qw & \gate[3]{\cP^k}  \qw & \gate{\cQ_1^{\operatorname{out}}}& \gate[3]{\Phi_2}&\gate{\cC_1}& \meter   \cw\\
\lstick{$\ket{0}$} & \gate{\mathcal{C}_2} \qw  & & \gate{\cQ_2^{\operatorname{in}}}\qw & & \gate{\cQ_2^{\operatorname{out}}}& &\gate{\cC_2}& \meter   \cw\\
\lstick{$\ket{0}$} & \gate{\mathcal{C}_3} \qw  & & \gate{\cQ_3^{\operatorname{in}}}\qw & & \gate{\cQ_3^{\operatorname{out}}}& &\gate{\cC_3}& \meter   \cw
\end{quantikz}}
    \caption{Structure of our noisy circuits for $3$ qubits in the presence of SPAM.}
    \label{fig:robust_parallelization_noisy}
\end{figure}
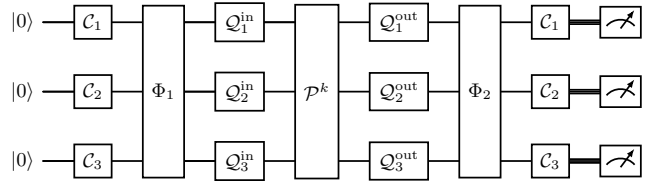
We now have:
\begin{prop}\label{prop:mean_variance}
Let $\Omega_k$ be the random variable defined in Equation \ref{eq:RV}, where we draw $\cC,Q^{\operatorname{in}}$ and $Q^{\operatorname{out}}$ uniformly and independently at random. Then
\begin{align}
\mathbb{E}[\Omega_k]=C_P\,\alpha_P^k
\end{align}
and 
\begin{align}
\mathbb{V}(\Omega_k)\leq {3^{w(P)}}\,,
\end{align}
where $C_P$ is given by 
\begin{align}
C_P\coloneqq 2^{-2n}\tr{\Phi_1^*(P)P}\tr{\Phi_2^*(P)P}.
\end{align}
\end{prop}
\begin{proof}
We will first fix the values of $Q^{\operatorname{in}}$ and $Q^{\operatorname{out}}$ and consider only the expectation values w.r.t. $\cC$. Note that $\omega_P(\cC,X_k)$ is only nonzero if $\cC$ maps the computational basis to the eigenbasis of $P$, a fact we will use repeatedly. As the local Cliffords are uniformly random, this happens with probability $3^{-w(P)}$. Let $P_+,P_-$ be the projectors onto the positive and negative eigenvalues of $P$. Again by the uniformity of the Cliffords, conditioned on the fact that we mapped to the eigenbasis of $P$ under the Cliffords, we prepare the initial state $\tfrac{P_+}{2^{n-1}}$ with probability $1/2$ and $\tfrac{P_-}{2^{n-1}}$ with probability $1/2$. We then evolve this state by $\Phi_2\cQ^{\operatorname{out}}\cP^k\cQ^{\operatorname{in}}\Phi_1$ and measure the POVM $\{P_+,P_-\}$. By the definition of the random variable $\omega_P\equiv \omega_P(\cC,X_k)$ we have that its expectation value conditioned on $\cC$ mapping onto the eigenbasis of $P$ is given by 
\begin{align*}
\mathbb{E}[\omega_P|\cC]&=3^{w(P)}2^{-n}\Big(\tr{P_+\Phi_2\cQ^{\operatorname{out}}\cP^k\cQ^{\operatorname{in}}\Phi_1(P_+)}\\
&\qquad \qquad +\tr{P_-\Phi_2\cQ^{\operatorname{out}}\cP^k\cQ^{\operatorname{in}}\Phi_1(P_-)}\\
&\qquad \qquad  -\tr{P_-\Phi_2\cQ^{\operatorname{out}}\cP^k\cQ^{\operatorname{in}}\Phi_1(P_+)}\\
&\qquad \qquad -\tr{P_+\Phi_2\cQ^{\operatorname{out}}\cP^k\cQ^{\operatorname{in}}\Phi_1(P_-)}\Big)\\
&=3^{w(P)}2^{-n}\tr{P\Phi_2\cQ^{\operatorname{out}}\cP^k\cQ^{\operatorname{in}}\Phi_1(P)}
\end{align*}
As the probability of hitting the right basis is $3^{-w(P)}$, it will cancel out the $3^{w(P)}$. Let us now analyze the effect of the additional Pauli twirling.
By definition of the Twirling (see Eq.~\eqref{equ:twirling_Pauli}), if we take the expectation value over $Q^{\operatorname{in}}$ and $Q^{\operatorname{out}}$ now, we obtain:
\begin{align}
 \mathbb{E}[\Omega_k] =  2^{-n}\tr{\cT_P^*(\Phi_2^*(P)) \cP^k\circ \cT_P(\Phi_1(P))}.
\end{align}
Note that $\cT_P^*=\cT_P$ and, thus
\begin{align}
 \cT_P^*(\Phi_2^*(P))=2^{-n}\tr{\Phi_2^*(P)P}P
\end{align}
and similarly
\begin{align}
\cT_P(\Phi_1(P))=2^{-n}\tr{\Phi_1(P)P}P=2^{-n}\tr{\Phi_1^*(P)P}P,
\end{align}
which proves our claim about the expected value. The estimate of the variance follows again from noting that the random variable $\Omega_k$ is nonzero only if $\cC$ maps the computational basis to an eigenbasis of $P$, which happens with probability $3^{-w(P)}$. In that case, the value of the random variable squared is bounded by $9^{w(P)}$, which gives that the expected value of $\Omega_k^2$ is bounded by $3^{w(P)}$. Our bound on the variance follows.
\end{proof}
Some remarks are in order: first, note that as $P$ is an eigenvector of $\cP$, $Q^{\operatorname{in}}$ has no effect in the protocol. However, this extra step ensures that small deviations from $\cP$ being a Pauli channel will be suppressed. Furthermore, it should be straightforward to obtain a similar protocol for the case where input and output Pauli matrices are not the same at the expense of a higher sample complexity in terms of the locality.

The following corollary follows directly from the joint use of Proposition \ref{prop:mean_variance} and a median of means argument:
\begin{cor}\label{ref:corSPAMlearning}
In the same setting as Prop.~\ref{prop:mean_variance}, let $P_1,\ldots,P_m$ be $m$ Pauli strings such that $w(P_i)\leq w$, denote by $p_0$ the probability that no error occurs, i.e.~$p_0=\mu(\mathds{I})$, and  
\begin{align}
2^{-2n}|\tr{\Phi_1^*(P_i)P_i}\tr{\Phi_2^*(P_i)P_i}|\geq C_{\operatorname{SPAM}}.
\end{align}
Then $\mathcal{O}(3^w\epsilon^{-2}C_{\operatorname{SPAM}}^{-2}\log(m\delta^{-1}))$ runs of circuits of the form of Fig.~\ref{fig:robust_parallelization_noisy} suffice to estimate all $\alpha_{P_i}$ up to an error $\epsilon(1-p_0)$ with probability of success at leat $1-\delta$. Furthermore, the maximal value of $k$ needed in each run satisfies $k=\mathcal{O}(\log(1/(1-p_0)))$.
\end{cor}
\begin{proof}
From Prop.~\ref{prop:mean_variance} it follows that using a median of means estimator~\cite{Devroye2016} we can estimate the mean of all random variables with expectation value $C_{P_i}\alpha_{P_i}^k$ for a given $k$ up to $C_{\textrm{SPAM}}\epsilon$ and probability of success at least $1-\delta$ with the advertised sample complexity. We can then use the results of~\cite{PhysRevA.99.052350} to fit these values to an exponential curve and obtain an estimate that satisfies the error bound $\epsilon(1-p_0)$ by using the channel $k=\mathcal{O}(-\log(1-p_0))$ times.
\end{proof}
Thus, by combining the randomized experiments with an exponential fitting, it is possible to recover all Pauli coefficients in a SPAM-robust way.
Note that $C_{\textrm{SPAM}}$ will usually also scale exponentially in the locality of the Pauli string for local noise models. E.g., for local depolarizing noise with depolarizing probability $q$, $C_{\textrm{SPAM}}\geq(1-q)^{2w}$.

\end{document}